\documentclass[11pt,journal,compsoc]{IEEEtran}
\usepackage[utf8]{inputenc}
\usepackage[T1]{fontenc}

\usepackage{graphicx}
\usepackage{amsmath, amsfonts, amssymb, amsthm}
\usepackage{algorithmic}
 
\usepackage{color, colortbl} 
 
\newtheorem{theorem}{Theorem} 
\newtheorem{lemma}{Lemma}

\newtheorem{cor}{Corollary}

\author{
Nathan~Davidov,     
Amanda~Hernandez,
Justin~Jian,        
Patrick~McKenna,
K.A.~Medlin, 
Roadra~Mojumder,
Megan~Owen,         
Andrew~Quijano,     
Amanda~Rodriguez,   
Katherine~St.~John, 
Katherine~Thai,     
Meliza~Uraga        

\IEEEcompsocitemizethanks{
\IEEEcompsocthanksitem N.~Davidov \& R.~Mojumder: Department of Mathematics, Hunter College, City University of New York (CUNY), New York, NY, USA, \texttt{nd2017@hunter.cuny.edu} \& 
\texttt{mroadra@gmail.com}.
\IEEEcompsocthanksitem
A.~Hernandez: University of Connecticut, Bridgeport, CT, USA,
\texttt{amanda.hernandez@uconn.edu}.
\IEEEcompsocthanksitem
J.~Jian:  
  \texttt{justin.jian@nyu.edu}
\IEEEcompsocthanksitem
P.~McKenna:
\texttt{patrick.mckenna@student.fairfield.edu}
\IEEEcompsocthanksitem
K.~Medlin, Department of Mathematics, University of North Carolina, Chapel Hill, NC, USA,
\texttt{kmedlin@unc.edu}.
\IEEEcompsocthanksitem M.~Owen: Department of Mathematics, Lehman College, CUNY, Bronx, NY, USA,
\texttt{megan.owen@lehman.cuny.edu}.
\IEEEcompsocthanksitem
A.~Quijano: New York University, New York, NY, USA,
 \texttt{Andrew.Quijano@columbia.edu}.
\IEEEcompsocthanksitem A.~Rodriguez: 
\texttt{rodriguezamanda214@gmail.com}.
\IEEEcompsocthanksitem
K.~St.~John: ACM Member,
Department of Computer Science, Hunter College, CUNY \& Invertebrate Zoology, American Museum of Natural History, New York, NY, USA,
\texttt{katherine.stjohn@hunter.cuny.edu}.
\IEEEcompsocthanksitem
K.~Thai:
College of Information and Computer Sciences, University of Massachusetts, Amherst, MA, USA
  \texttt{kbthai@umass.edu}.
\IEEEcompsocthanksitem
M.~Uraga:
  \texttt{melizaura3050@gmail.com}.
\IEEEcompsocthanksitem
This work was partially supported by the US National Science Foundation Research Experience for Undergraduates Program, Grant \#1461094, CAREER grant DMS-1847271, and grants from the Simons Foundation.\protect\\
}}

\title{{Maximum Covering Subtrees for Phylogenetic Networks}}

\begin{document}

\maketitle

\begin{abstract}
Tree-based phylogenetic networks, which may be roughly defined as leaf-labeled networks built by adding arcs only between the original tree edges, have elegant properties for modeling evolutionary histories.  We answer an open question of Francis, Semple, and Steel about the complexity of determining how far a phylogenetic network is from being tree-based, including non-binary phylogenetic networks. We show that finding a phylogenetic tree covering the maximum number of nodes in a phylogenetic network can be computed in polynomial time via an encoding into a minimum-cost flow problem.
\end{abstract}


\section{Introduction}

A fundamental question in biology is to determine the evolutionary history of a set of species.  Many evolutionary histories can be captured by trees with leaves labeled by the species in question, but in some cases, a tree structure does not capture evolution's complexity.  Instead, phylogenetic networks better models the multiple paths in the history. 
Many questions are computationally hard on the class made up of all phylogenetic networks, including determining if a network displays a fixed tree \cite{kanj2008} and counting the number of trees displayed by the network \cite{linz2013}.  There has been much work to find restricted classes of phylogenetic networks that are both complex enough to model the underlying biological phenomena and small enough to allow tractable computations \cite{vanIersel2010}.

Francis and Steel \cite{francis2015} developed a natural class of phylogenetic networks, called \emph{tree-based phylogenetic networks} or 
\emph{tree-based networks}, that consists of networks that can be represented as a tree with additional edges added only between the original tree edges.  
Tree-based phylogenetic networks capture the idea that a phylogenetic network can have an underlying tree as its ``backbone.''
They gave a polynomial algorithm to determine if a phylogenetic network is tree-based.  Anaya {\em et al.} \cite{anaya2016} showed that deciding if a phylogenetic network is based on a specific tree is NP-hard.  
Semple \cite{semple2016} showed that the class of tree-child networks is precisely the class of tree-based phylogenetic networks with the property that every embedded phylogenetic tree is a base tree.  Zhang \cite{zhang2016} classified tree-based phylogenetic networks in terms of matchings on a bipartite graph, which was later used by Francis {\em et al.} \cite{francis2018} to calculate some alternative measures of how far a phylogenetic network is from being tree-based.  The idea of tree-based phylogenetic networks has been extended to non-binary phylogenetic networks by Jetten and van Iersel \cite{jetten2016} with further work by Pons et al. \cite{pons2019}; to unrooted binary phylogenetic networks by Francis et al. \cite{francis2017tree}; and to unrooted, non-binary phylogenetic networks by Hendriksen \cite{hendriksen2018}.  

\begin{figure}
    \centering

    \includegraphics[height=1.325in]{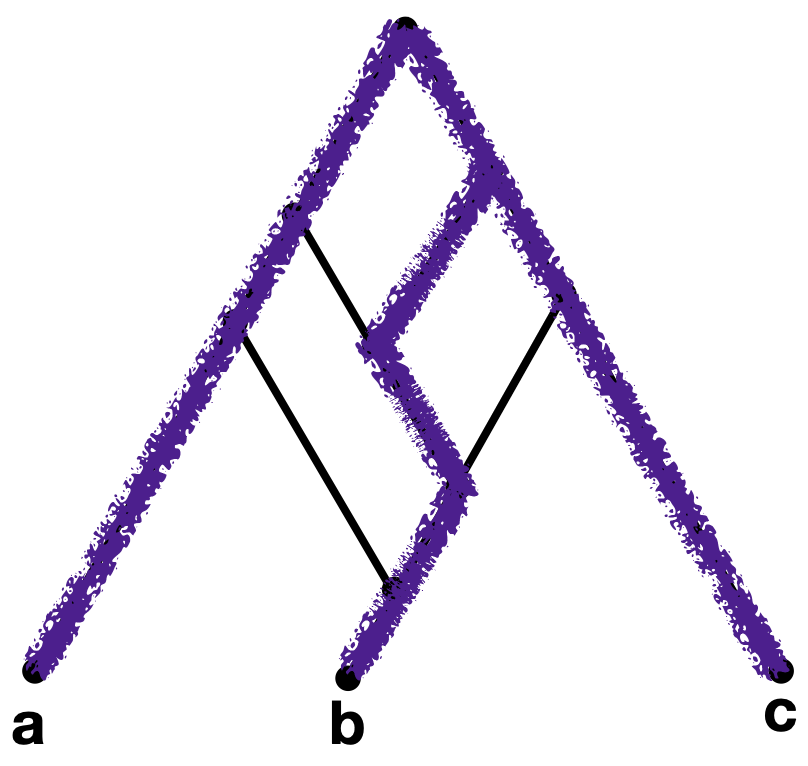} 
    
    \caption{A tree-based network.  The root is drawn at the top, and all edges are directed from top to bottom.  The base tree is indicated by shaded lines in the network.  }
    \label{fig:treeBased}
\end{figure}

Given the properties of the class of tree-based phylogenetic networks, a natural question is how far is a given network from being tree-based.
Francis {\em et al.} \cite{francis2018} proposed several ways to measure how much a phylogenetic network deviates from being tree-based. They showed that finding the minimal number of leaves that need to be attached to a phylogenetic network to make it tree-based can be computed in polynomial time. They left as an open question what is the fewest number of vertices that must be absent
{to yield a tree-based network}. Our contribution is to show that this latter measure can be computed in polynomial time via a network flow approach.  This approach works for both the traditional binary phylogenetic networks, as well as networks where higher degree nodes (i.e. non-binary networks) are allowed.  Finding a tree covering the maximum number of nodes in a phylogenetic network is equivalent to finding a set of disjoint paths, each ending in a network leaf, that cover as many nodes as possible {\cite{francis2018}}. This problem can be solved by encoding the phylogenetic network into a flow network with costs associated to each arc.  We use a standard node splitting approach \cite{ahuja1993}, which has also been used to solve the minimum path cover problem on directed acyclic graphs, to encode the problem as a minimum-cost flow problem.

Solving the minimum-cost flow problem on this network can be {computed} in polynomial time, and we show that this optimal flow yields a maximum covering subtree of the original network {that can be computed in polynomial time}.

Our problem falls into an interesting middle ground of complexity results about covering paths and trees.  Fortune {\em et al.} \cite{fortune1980} showed that finding a set of $k$ disjoint paths that cover the vertices of a directed graph is NP-complete.  Perl and Shiloach \cite{perl1978} proved that it can be computed in linear time when $k=2$.  Seminal work of Robertson and Seymour showed that finding $k$ vertex disjoint paths for undirected graphs is solvable in polynomial time \cite{robertson1995}.  Efficiently computing vertex-disjoint paths and out-branchings (spanning trees in a directed graph that start at a specified root) has much intriguing work 
(e.g.~\cite{alon2009,eppstein1998,gutin2009,suurballe1984,tholey2005}) that has combined classic results (such as flow across bipartite graphs of matchings) with elegant data structures.
{T}he closely related question of partitioning a directed acyclic graph (DAG) into as few strongly connected components possible is {NP}-hard \cite{leskovec2009}.

\section{Background}

Our basic structures are leaf-labeled phylogenetic trees and networks.
A {\em phylogenetic network} or {\em network} $N$ is a connected, directed, acyclic graph denoted by the tuple $(V, A, L, \rho)$ where $V$ is a set of {\em vertices}, 
$A$ is the set of {\em arcs}, or directed edges, 
$L$ is the set of labels on the {\em leaves} (vertices with out-degree $0$ {and in-degree $1$}), and 
$\rho$ is the root (the unique vertex with in-degree $0$).  We let $k=|L|$, $n = |V|$, $m = |A|$.
The remaining vertices can be placed in two groups:  {\em tree vertices}, which have in-degree $1$ and out-degree greater than or equal to $2$; and {\em reticulation vertices}, which have in-degree greater than or equal to $2$ and out-degree $1$.  If the in-degrees and out-degrees are bounded by $2$, we refer to the network as {\em binary}.  If not, the network is called {\em non-binary}.  If a phylogenetic network has the further restriction that all vertices have in-degree 1 (or 0 in the case of the root), it is called a {\em tree}. {We focus on subtrees of networks that contain all the leaves and the root and refer to these as {\em covering subtrees}. We note that if all vertices with in-degree $1$ and out-degree $1$ of a tree are surpressed, the resulting structure is called a phylogenetic tree.  }

\subsection{Tree-Based Networks}

Following {\cite{francis2015} and \cite{jetten2016}}, a phylogenetic network $N = (V, A, L, \rho)$ is {\em tree-based} if there exists a tree $T = (V', A', L, \rho)$ where the root and leaf set of $T$ are identical to those of $N$, the vertices are identical (i.e. $V' = V$), and arcs are subsets (i.e. $A' \subseteq A$). 
{$T$ is referred to as the {\em base tree} for the network $N$.} 
{Alternatively, given a phylogenetic tree $T$, we can add additional vertices to its edges and arcs between these new vertices to yield the network $N$}.  Note that some of the vertices in the tree may only have degree two within the tree, since the third in-coming or out-going arc is part of the network, but not the tree.  
Figure~\ref{fig:treeBased} shows a tree-based network.  By convention, we draw the root (the distinguished node with in-degree 0) at the top of the figure and directional arrows are left out since all edges are directed downwards.  {A} base tree is indicated by shaded lines in the figure.  We will also look at networks that are not tree-based (see Figure~\ref{fig:encoding}{(a)}).  For those, by definition, there is no base tree, but there are subtrees that have the same root and leaf set and contain a subset of the vertices and directed edges.
Formally, for a network, $N = (V, A, L, \rho)$, we define a {\em covering subtree} as 
a connected tree $T = (V', A', L, \rho)$ that contains the root and all of the leaves of the network $N$, subsets of the vertices and directed edges.  We emphasize that for $T$ to be a covering tree of $N$, all leaves of $N$ must also be leaves of $T$.  A {\em maximum covering subtree} of a network $N$ is a covering subtree of $N$ that contains the most vertices of $N$, or equivalently, maximizes $|V'|$.
Figure~\ref{fig:encoding}(a) shows a phylogenetic network that is not tree-based.  A maximum covering subtree of this network is indicated by the shaded lines in Figure~\ref{fig:encoding}(d).

As an open question, Francis et al. \cite{francis2018} asked what is the complexity {for computing} the following measure:
Given a network $N$ and a covering subtree $T$ of $N$,
let $\eta(N, T)$ be the number of vertices in $N$ that are not covered by $T$.  That is, 
$$
    \eta({\cal N}, T) = |V({\cal N})\setminus V(T)|
$$

We let $\eta({\cal N})$ be the minimal number of uncovered vertices over all covering subtrees of ${\cal N}$
$$
    \eta(N) = \min_{T} \eta(N, T)
$$
We note that 
$\eta(N) = 0$ when $N$ is a tree-based network.
As with the other measures proposed by Francis {\em et al.} \cite{francis2018}, we show that this one {can be computed in} polynomial time.  We formalize the problem as:

\bigskip
\noindent
\textsc{\bf Maximum Covering Subtree (\textsc{Max-CST})}

\noindent
\textsc{Instance:}  A leaf-labeled network $N = (V, A, L, \rho)$ with vertices $V$, directed edges $A$, leaves $L$, and root $\rho$.

\noindent
\textsc{Question:} What is the minimum number of vertices {of $N$} not covered by a covering subtree of $N$, $\eta(N) = \min_T |V(N)\setminus V(T)|$, where $T$ ranges over all covering subtrees of $N$?

Note that in the case of unrooted phylogenetic networks, determining if the network is tree-based is NP-complete \cite{francis2017tree}.  Since finding a maximum covering subtree for an unrooted phylogenetic network would also determine if the network is tree-based, {the unrooted variant of MAX-CST} must be NP-hard.

\subsection{Network Flows}

To find the \textsc{Max-CST} for a phylogenetic network, we use the well-studied approach of network flows  \cite{ahuja1993}.
{We give a brief overview of flow networks, and, in the next section, we outline our encoding of phylogenetic networks into flow networks }.
We represent a flow network as $F = (V,A,s,t,c,w)$ where $V$ is a set of vertices, { $A = V\times V$} is a set of directed arcs, $s\in V$ is the source of the network and has in-degree $0$, $t\in V$ is the sink of the network and has out-degree $0$, $c:A\rightarrow \mathbb{R}^+$ is a function that assigns capacities to each edge, and $w:A\rightarrow \mathbb{R}$ is a function that assigns costs to each edge.
The general goal is to maximize the ``flow'' through the network where the flow on any arc cannot exceed its capacity and flow into each vertex matches the flow out of the vertex {(except for the source, $s$, and sink, $t$)}.    The flow across a network equals the flow into the sink.  To find the \textsc{Max-CST}, 
we use a standard variation called the \emph{minimum-cost flow problem}, in which the flow problem is modified {to} include a cost on each arc and the solution is a minimum-cost flow.  Specifically, for each arc, {$u\rightarrow v$}, we are given a cost, $w(u,v) \in \mathbb{R}$.  If the flow across the edge is $f(u,v)$ {in the network}, the cost of the flow is:
$$
    \sum_{u\rightarrow v \in A} w(u,v)\cdot f(u,v)
$$
We note that the costs of edges can be negative. {While flow networks can have varied capacity on arcs, each arc in our networks has the capacity of either $0$ or $1$.  The former, capacity set to $0$, is a technicality since by definition all possible arcs are included in a network, but those that are not used have zero capacity.  The latter, capacity set to $1$, is called `unit capacity'}.  
A standard framework for approaching network flow problems is the following: given a flow on a network, an augmenting path is an additional path with positive flow that can be added to that flow.  An additional ``residual'' network that keeps track of the current capacities of the flow is employed to make computing augmenting paths easier.  To find the minimum-cost augmenting paths, we use an algorithm that can handle negative weights, such as Bellman-Ford method (\cite{bellman1958,ford1956}) which runs in $O(|A|\cdot|V|) = O(mn)$.  To compute the minimum-cost flow, we follow the standard approach of incorporating the minimum-cost augmenting path from the residual network into the flow, recomputing the residual network, and repeating until no additional augmenting paths exist.  

Given the importance of the minimum-cost flow problem, there are numerous approaches to solving it, including combinatorial algorithms and linear programming \cite{ahuja1993}.  This problem continues to be an active area of research, including when restricted to unit capacities, as in our problem \cite{cohen2017negative, goldberg2017minimum, liu2020faster}.

\begin{figure*}
    \centering
    \begin{tabular}{ccccccc}
    \\
    \includegraphics[height=2.5in]{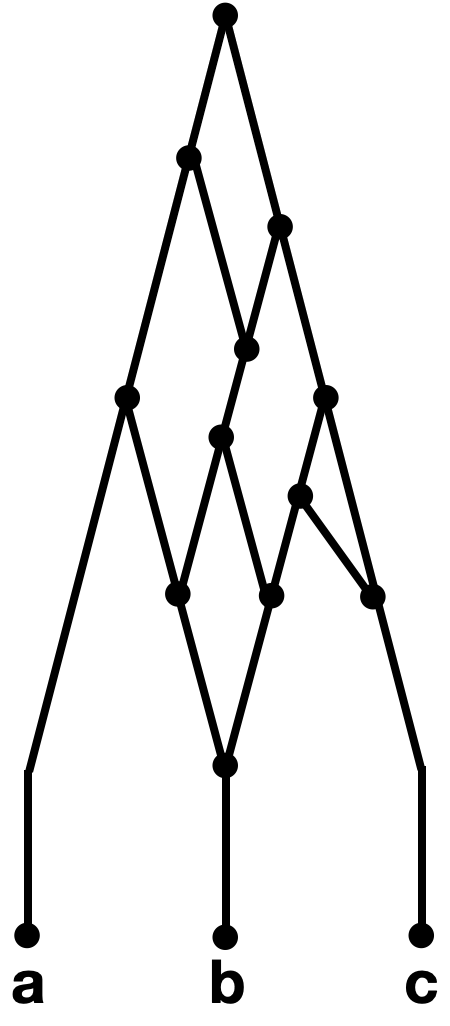}&
    \mbox{\hspace{.25in}}&
    \includegraphics[height=2.75in]{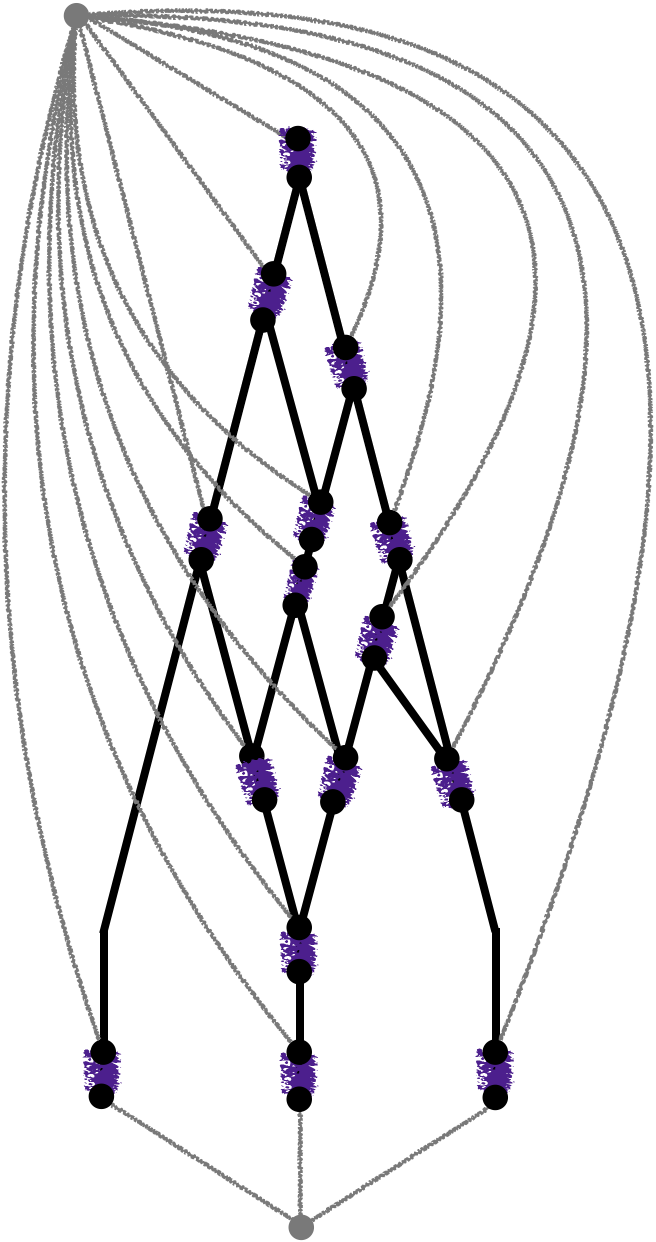}&
    \mbox{\hspace{.25in}}&
    \includegraphics[height=2.75in]{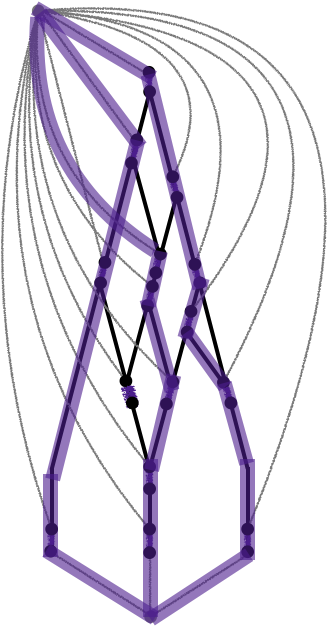}&
    \mbox{\hspace{.25in}}&
    \includegraphics[height=2.5in]{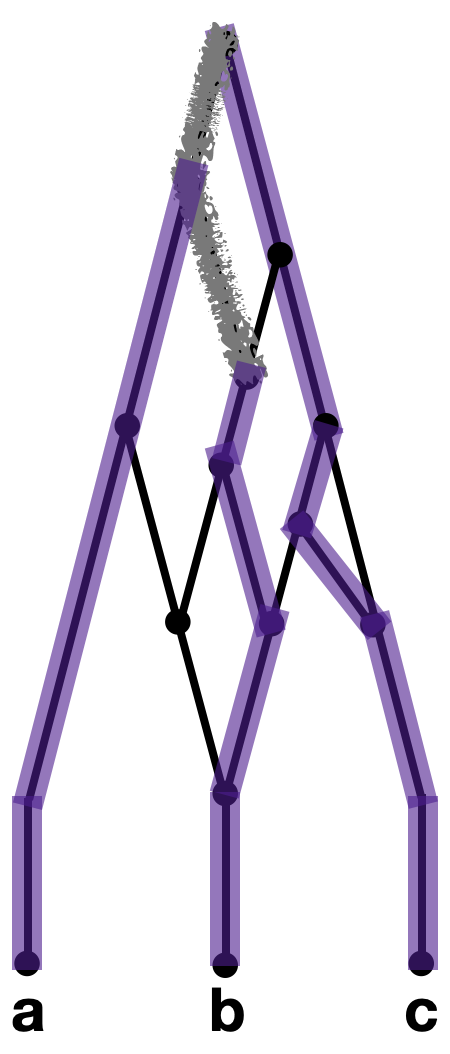}\\
    \\
    (a) && (b) && (c) && (d)\\
    \end{tabular}

    \caption{(a) A phylogenetic network and (b)  its corresponding flow network.  Each vertex $v$ in the network has been replaced by two vertices $(v_{in},v_{out})$ and connected with a ragged-edged (purple) arc weighted with cost $-1$. The remaining gray and black arcs have cost $0$.  {Only arcs with non-zero capacity are shown; all shown have capacity $1$.}  (c) A minimum-cost flow highlighted in the network and (d) the corresponding maximal covering subtree highlighted in the phylogenetic network.  The highlighted (purple) edges are directly copied from the flow network and are vertex disjoint paths ending in the $k$ leaves. 
    An additional $k-1$ rough-shaded (gray) edges connect the paths into a maximal covering subtree.
    }
\label{fig:encoding}
\end{figure*}

\section{Computing \textsc{Max-CST}s}
\label{sec:encoding}

Our approach breaks down into three steps: encoding our problem as a network flow, finding the minimum-cost flow for that network, and building a tree from the resulting flow.  We detail the steps below and prove correctness in Section~\ref{sec:proofs}.

\vspace{.1in}
\noindent
\textbf{Encoding as a flow network:} To use a flow network to find maximal paths, we follow the classic set up with a slight twist. 
Given a phylogenetic network, $N = (V, A, L, \rho)$, we build a flow network $F = (V',A',s,t,c,w)$.
First, we add an additional node, $s$, to serve as the source, and add arcs from $s$ to every node of the network.  We next add an additional node, $t$, to serve as the {sink}.  We add arcs from each leaf (nodes with out-degree $0$) to $t$.
Next, each vertex $v\in V$ from the original network is replaced with a pair $(v_{in}, v_{out})$ that are connected by an arc from $v_{in}$ to $v_{out}$.  All incoming arcs to $v$ are redirected to $v_{in}$.  All outgoing arcs from $v$ are redirected to start from $v_{out}$.  We let $A'$ be the arcs described above, and 
$V' = \bigcup_{v\in V} \{v_{in},v_{out} \} \cup \{s,t\}$.  The capacity function, $c: A' \rightarrow \mathbb{R}^+$ is 1.  The cost function $w: A' \rightarrow \mathbb{R}$ is $-1$ for the arcs $v_{in} \rightarrow v_{out}$, $v\in V$, and 0 for all other arcs.  
The new network is a directed acyclic graph with costs on the edges of either $0$ or $-1$. 
Figure~\ref{fig:encoding}(a) shows a phylogenetic network and Figure~\ref{fig:encoding}(b) shows the corresponding flow network.  Two additional vertices for the source and {sink} are added at the top and bottom of the figure, respectively.  The original network arcs are in solid thin lines.  The additional arcs, $v_{in}\rightarrow v_{out}$, are rough-edged, and the additional arcs from the source to all $v_{in}$ vertices and the arcs from all $v_{out}$ of the leaves to the target are in gray.  

We note that the encoding takes linear time in the number of arcs of the phylogenetic network, so, is $O(m)$.

\vspace{.1in}
\noindent
\textbf{Computing Minimum-Cost Flow:}
We employ the standard approach to finding the minimum-cost  flow across the encoded network \cite{ahuja1993} that sets the initial flow across the network to $0$ and then repeatedly incorporates the minimum-cost augmenting path until no more can be added.  
For all such networks, there is a minimal-cost  solution with integral flows across the edges, and we will use this fact to prove correctness of our approach below.  
Since the costs on arcs can be negative, the approach used to find minimum-cost augmenting paths must be able to handle negative costs.  One such method is the Bellman-Ford approach which runs in $O(nm)$.  By construction of the encoding, there are at most $k$ disjoint paths between the source and the sink (bounded by the number of arcs that are directed into the sink).  The approach takes $O(n^2m)$ time to give a minimal-cost flow (highlighted in Figure~\ref{fig:encoding}(c)).

\vspace{.1in}
\noindent
\textbf{Building the Tree from the Flow:}  The encoded network is constructed to yield, as the minimal-cost flow result, a solution of $k$ disjoint paths, $p_1,\ldots,p_k$, that maximally cover the original network (see \cite{ahuja1993} for proofs).  We will show that $k = |L|$, the number of leaves in the network.  Without loss of generality, assume that $p_1$ contains the root of the network.  Let $r_i$ be the first vertex (i.e. the vertex with in-degree $0$) of path $p_i$, $i=1,\ldots,k$.  To build a covering subtree from these disjoint paths, we show in Section~\ref{sec:proofs} that for each path, $p_i$, $i=2,\ldots,k$, a parent of $r_i$ is a vertex in another path.  We then add an arc from a parent of $r_i$ directed towards $r_i$.  We note that $r_i$ could be a reticulation node and have two parents, but choosing either will result in a tree.  Figure~\ref{fig:encoding}(d) shows the $k$ disjoint paths  and the $k-1$ additional edges added to form a tree.  

Building the tree from the flow network takes linear time in the 
{number of vertices and number of the paths, or $O(n+m) = O(m)$}.

\section{Proofs}
\label{sec:proofs}

For this section, we assume that $N = (V, A, L, \rho)$ is a phylogenetic network that has been encoded into a flow network $F = (V',A',s,t,c,w)$ via the procedure in Section~\ref{sec:encoding}.  Let $f$ be a minimal-cost flow on $F$.  Let $f|_N$ be the restriction of the flow $f$ to the original network, in that:
$$
    f|_N = \{u\rightarrow v |  
        \mbox{$f(u_{out},v_{in}) > 0$ and $u\rightarrow v \in A$} \}
        \subseteq A.
$$
We first {note that for networks with unit capacities, there exists a minimal cost flow $f$ with integral values \cite{ahuja1993}.}
We show that the restriction of the minimum-cost flow to $N$, $f|_N$ consists of disjoint paths of $N$, each containing a leaf of the original network:

\begin{lemma}
$f|_N$ consists of $k$ disjoint paths, $p_1,\ldots,p_k$ of $N$.  
Further, each path contains exactly one leaf of $L$.
\label{lem:paths}
\end{lemma}

\begin{proof}
We have that $f$ is a minimal cost flow and $f|_N$ {is the arcs of $f$ that are also arcs of $N$}. 
If $f|_N$ is not the union of disjoint paths, then there exists a vertex, $x$ {that is the starting point of two or more arcs of $f|_N$ or the ending point of two or more arcs of $f|_N$. Assume that the former is true and that there are $k>1$ arcs of $f|_N$ with $x$ as a starting point called $x\rightarrow u_1$, $x\rightarrow u_2,\ldots x\rightarrow u_k$.  Since these arcs are in $f|_N$, each has non-zero flow.}
Since we are assuming that the minimal cost flow is integer-valued, and the capacity of the arc {exiting} $x$ is 1, at most one of these {outgoing} arcs have non-zero flow, contradicting the existence of $x$.  A similar argument applied to latter case.

$f$ must contain all the arcs $v_{in} \rightarrow v_{out} $ for $v \in L$, since if it does not, the flow can be augmented by a path from the source to the missing leaf vertices to the {sink},
$s \rightarrow v_{in} \rightarrow v_{out} \rightarrow t$.  Since all arcs $v_{in} \rightarrow v_{out}$ have cost $-1$, the new flow would have a lower cost, contradicting the minimality of $f$.
\end{proof}

To build the tree from the flow, we show that for each path either the starting vertex of the path is the root of the network, or that the parents of the starting vertex of the path belong to another path:

\begin{lemma}
Let $p$ be a path of $f|_N$ and $r$ be the vertex of $p$ with in-degree $0$ in $p$.  Then either $r = \rho$, the root of $N$, or there exists another path $p'$ that contains a vertex $v$ such that $v\rightarrow r$ is an arc of $N$.
\label{lem:parentPaths}
\end{lemma}

\begin{proof}
If $r = \rho$ we are done, so assume not.  To prove a contradiction, further assume that for $r$, none of its incoming arcs in $N$ are in $f|_N$. So, $f(s,r_{in}) = 1$ and there exists $v \in V$ with $v \rightarrow r \in A$ and $v$ not {a vertex of an arc of} $f|_N$, which implies $f(v_{in},v_{out}) = 0$.  Since the cost, $w(v_{in},v_{out}) = -1$, we can augment the network flow with $s \rightarrow v_{in} \rightarrow v_{out} \rightarrow r_{in}$, replacing the $s \rightarrow r_{in}$ arc.  This new network has lower cost, yielding a contradiction.
\end{proof}

Using Lemma~\ref{lem:parentPaths}, we can connect the disjoint paths of $f|_N$ into a covering subtree of $N$:

\begin{lemma}
There exists a tree $T$ on leaves $L$ that contains all the arcs of $f|_N$.  Further, the vertices of $T$ are exactly {
$\{v | \mbox{ $\exists w$,  $(v_{out},w_{in})\in f|_N$ or 
$(w_{out},v_{in})\in f|_N$}\}$. }
\label{lem:inducedTree}
\end{lemma}

\begin{proof}
By Lemma~\ref{lem:paths}, $f|_N$ consists of $k$ disjoint paths, each containing exactly one leaf of $L$.  By Lemma~\ref{lem:parentPaths}, the start of each path, $r_i$, $i = 1,\ldots,k$, is either the root $\rho$ or we can associate a parent $u_i$ {such that $u_i$ is a vertex of another path $p_j$ , $j\ne i$.}  Since the paths are disjoint, $p_i = \rho$ in only one path.  Without loss of generality, assume that $p_1$ has the root as its starting vertex.  Let the tree $T$ be $f|_N$ with the arcs $u_i \rightarrow r_i$, $i = 2,\ldots,k$ added.  We note that this construction of connecting a path at its start to a tree does not introduce undirected cycles {maintaining} the tree property of the construction.
\end{proof}

To show that the constructed subtree is a maximal covering subtree, we need first to show that:

\begin{lemma}
\label{lemma:backwards}
Every {covering subtree} $T$ in $N$ corresponds to a flow in the flow network $F$ with cost $-|U|$ where $U$ is the set of vertices of $T$.
\end{lemma}

\begin{proof}
Let $T$ be a {covering subtree} of the network $N = (V,A,L,\rho)$, {and $F=(V',A',s,t,c,w)$ be its flow network}.  Since it is a tree on $k$ leaves, there exists {at most} $k-1$ vertices with out-degree greater than or equal to $2$.  For each of those vertices, delete {all but one of the} outgoing arc to create $k$ disjoint paths $p_1,p_2,\ldots,p_k$.  Let $r_i$ be the vertex of $p_i$ with in-degree $0$ in $p_i$ and let $l_i$ be the vertex of $p_i$ with out-degree $0$ in $p_i$.  Define $f:A' \rightarrow \{0,1\}$ as follows: $f(u,w) = 1$ 
if 
\begin{itemize}
    \itemsep 0pt
    \item $u = v_{in}$ and $w = v_{out}$ and $v$ is a vertex of some path $p_i$, $i=1,\ldots,k$, 
    \item $u = v_{out}$ and $w = v'_{in}$ and $v \rightarrow v'$ is an arc in some $p_i$, $i = 1, ..., k$,     
    \item $u = s$ and $w = (r_i)_{in}$ for some $i=1,\ldots,k$, or
     \item $u = (l_i)_{out}$ 
     {and $w=t$} for some $i=1,\ldots,k$.  
\end{itemize}
For all other arcs, set the $f$ to $0$.
Since the flow $f$ is either $0$ or $1$, it never exceeds the capacity on an arc. 
By construction, for every vertex $v\not\in \{s, t\}$, the flow into an arc equals the flow out of an arc. Thus, $f$ is a flow on the network $F$. 

Furthermore, flow $f$ is only positive on arcs $v_{in} \rightarrow v_{out}$ if $v$ is a vertex in some path $p_i$, $i = 1, ..., k$.  Since the arcs $v_{in} \rightarrow v_{out}$ for $v \in V$ are the only non-zero cost arcs, and each have cost $-1$, the total cost of $f$ is $-|U|$.
\end{proof}

Finally, we show that the minimum-cost flow induces a maximal covering subtree of the phylogenetic network $N$:

\begin{theorem}
\label{thm:maxTree}
{Let $N$ be a phylogenetic network and let $F$ be the corresponding flow network.}
Let $f$ be a minimum-cost flow on flow network $F$, and let $T$ be a {covering subtree of} $N$ as constructed in Lemma~\ref{lem:inducedTree}.  Then tree $T$ is {maximum covering subtree of $N$}.
\end{theorem}

\begin{proof}
By contradiction.  Assume not, then there is a tree, $T'$ that covers more vertices than $T$.  By Lemma~\ref{lemma:backwards}, there exists flow $f'$ that corresponds to $T'$ with cost $-|U'|$, where $U'$ is the set of vertices of $T'$.  Similarly, flow $f$ has cost $-|U|$ where $U$ is the set of vertices of $T$.  But $-|U'| < -|U|$ since $|U| < |U'|$, contradicting the minimal cost of flow $f$.  
\end{proof}

{
We note that the running time is cubic, but running time improvements are possible by assuming additional structure on the underlying network \cite{hayamizu2018}.}
As a corollary to the theorem, we have:

\begin{cor}
Given a phylogenetic network $N$, 
the maximum covering subtree of $N$ and $\eta(N)$ can be computed in polynomial time.
\end{cor}

\section{Conclusion}

Francis and Steel \cite{francis2015} introduced the class of tree-based networks that capture the biological question of when can an evolutionary history be modeled by a tree, or tree with a few additional arcs.  They showed that membership in this class can be decided in polynomial time, and in subsequent work with Semple \cite{francis2018}, they proposed measures for how close a network is to being tree-based.  They show that several measures of how close a network is to being tree-based, including the minimum number of leaves that need to be attached to make the network tree-based, can be computed quickly. They left as an open question to compute how many vertices are left uncovered by a maximal covering subtree.  In this paper, we show that this latter measure {can be computed in} polynomial time to compute via a network flow approach for traditional phylogenetic networks as well as the extension to non-binary networks.

\subsection*{Acknowledgements.}

We would like to thank the American Museum of Natural History and the CUNY Advanced Science Research Center for hosting us for several meetings.  This work was funded by a Research Experience for Undergraduates (REU) grant from 
the U.S. National Science Foundation (\#1461094 to St.~John and Owen).  Owen and St.~John and were also partially supported by grants from the Simons Foundation, and Owen by a CAREER grant from the National Science Foundation (\#1847271).

\bibliography{phylo}

\end{document}